%% file: HuangS16_QIP17.extended_abstract.tex
\def\draft{1}  
\documentclass[11pt]{article}
\usepackage[letterpaper,hmargin=1in,vmargin=1in]{geometry}

\usepackage{ifthen}

\newcommand{\talkingPoint}[1]{\ifthenelse{\equal{\draft}{1}}{{\color{brown}{---#1---}}}{}}
\newcommand{\CJ}[1]{\ifthenelse{\equal{\draft}{1}}{{\color{red}{---CJ:#1---}}}{}}
\newcommand{\YY}[1]{\ifthenelse{\equal{\draft}{1}}{{\color{blue}{---YY:#1---}}}{#1}}
\newcommand{\commentout}[1]{}

\usepackage{hyperref}
\hypersetup{pdfpagemode=UseNone}

\usepackage{amsfonts}
\usepackage{amssymb}
\usepackage{amsmath}
\usepackage{latexsym}
\usepackage{amsthm}
\usepackage{color}
\usepackage{enumerate}

\newtheorem{theorem}{Theorem}

\newtheorem{lemma}[theorem]{Lemma}
\newtheorem{definition}{Definition}

\newtheorem{corollary}[theorem]{Corollary}

\newtheorem{theorem*}{Theorem}
\newtheorem{corollary*}{Corollary}

\def\Tr{\textnormal{Tr}}

\def\<{\langle}
\def\>{\rangle}

\def\C{\mathcal{C}}

\def\E{\mathcal{E}}

\def\negl{\text{negl}}

\numberwithin{theorem}{section}
\numberwithin{equation}{section}

\title{Quantum hashing is maximally secure against classical leakage}

\author{%
   Cupjin Huang and Yaoyun~Shi  \\
 \\
  Department of Electrical Engineering and Computer Science\\
  University of Michigan, Ann Arbor, MI 48109, USA\\
  \texttt{{cupjinh,shiyy@umich.edu}}
}

\date{}

\begin{document}
\maketitle
\thispagestyle{empty}

\begin{abstract}
\input{HuangS16_QIP17.abstract.tex}

\end{abstract}

\clearpage
\setcounter{page}{1}

\input{HuangS16_QIP17.intro.tex}

\newpage
\input{HuangS16_QIP17.technical_content.tex}
\newpage
\bibliographystyle{abbrv}
\bibliography{qhash}
\end{document}

%% file: HuangS16_QIP17.abstract.tex
\noindent
Cryptographic hash functions are fundamental primitives widely used in practice.
For such a function $f:\{0, 1\}^n\to\{0, 1\}^m$,
it is nearly impossible for an adversary to produce the hash $f(x)$ without knowing the secret message 
$x\in\{0, 1\}^n$. Unfortunately, all hash functions are vulnerable under the side-channel attack, which is a
grave concern for information security in practice. This is because typically $m\ll n$ and an adversary needs only $m$ bits of information to pass the verification test.

In sharp contrast, we show that when quantum states are used, the leakage allowed can be almost the entire secret.  More precisely, we call a function that maps $n$ bits to $m$ qubits a quantum cryptographic function
if the maximum fidelity between two distinct hashes is negligible in $n$.
We show that for any $k=n-\omega(\log n)$, all quantum cryptographic hash functions remain 
cryptographically secure when leaking $k$ bits of information.
By the quantum fingerprinting constructions of Buhrman et al.~({\em Phys. Rev. Lett.} 87, 167902), 
for all $m=\omega(\log n)$, there exist such quantum cryptographic hash functions. We also show that one only needs $\omega(\log^2 n)$ qubits to verify a quantum cryptographic hash, rather than the whole classical information needed to generate one.

Our result also shows that to approximately produce a small amount of quantum side information on a classical secret, it may require almost the full information of the secret. 
This large gap represents a significant barrier for proving quantum security of classical-proof extractors.

%% file: HuangS16_QIP17.intro.tex
\section{Introduction}
\subsection{The problem and the motivation}
Cryptographic hash functions are a fundamental primitive used widely in today's cryptographic systems. They are  considered ``workhorses of modern cryptography''\footnote{Bob Schneier, \url{https://www.schneier.com/essays/archives/2004/08/cryptanalysis_of_md5.html}.} For simplicity, we focus our discussions on {\em keyless} (cryptographic) hash functions, each of which is an efficiently computable function $h$ from some message space $\mathcal{M}$ to some {\em digest space} $\mathcal{T}$~\cite{BonehS:book,preneel1994cryptographic, rogaway2004cryptographic}. Ideally, we want the hash function to have the following properties. First, the digest should be much shorter than the message. Depending on applications, the following security properties are desirable~\cite{rogaway2004cryptographic}. (1) {\em Collision resistant}: It is computationally infeasible to find a ``collision'', i.e., two distinct messages $x$ and $x'$, such that $h(x)=h(x')$.  (2) {\em Preimage Resistance}: It is computationally infeasible to invert $h$. (3) {\em Second Preimage Resistance}: Given a message $x$, it should be computational infeasible to $x'\neq x$ with $h(x')=h(x)$.

Prominent examples of widely used cryptographic hash functions include SHA-256 and SHA-512, part of the SHA-2 algorithms that were designed by NSA and are US Federal Standards. These algorithms are used in UNIX and LINUX for secure password hashing, in Bitcoin for proof-of-work.
As a motivating example, we consider how proof-of-work can be carried out through hash. Suppose that Alice receives a trove of valuable documents $x\in\{0, 1\}^n$, and Bob claims that he was the person producing and sending it. To prove his claim, he  sends Alice a tag $t\in\{0, 1\}^m$, which supposedly is the result of applying a
cryptographic hash function $h:\{0, 1\}^n\to\{0, 1\}^m$ on $x$. Alice simply checks if $t=h(x)$. Accept if yes, reject otherwise. By the collision resistance property, it is nearly impossible that Bob can produce $h(x)$ without knowing $x$. 


In practice, there may be information leakage of the message over time due to information transmission, adversarial attacks, etc. Therefore, it is rather desirable if the hash function is resilient against information leakage. We ask: how many bits $\ell$ about the message $x$ can be leaked before the adversary is able to forge the tag $h(x)$ easily? 

Cleary, $\ell \le m$, since if the tag $h(x)$ itself is known to the adversary, he does not need to know more about $x$ to pass the verification. This is rather disappointing, since $m$ is typically much smaller than $n$. 
We then ask: what if a quantum tag is used instead? If the leakage is quantum, by the same reasoning, $m$ remains a trivial and rather lower upper-bound on $\ell$. This leads us to our central question: 
{\em Can a quantum hash function be much more resilient to {\em classical} leakage?}

\subsection{Quantum cryptographic hash functions}
By a ``quantum hash function,'' we simply mean a classical-to-quantum encoding $\phi: \{0, 1\}^n\to {\mathbb{C}}^{2^m}$ that maps $x\in\{0, 1\}^n$ to a pure $m$-qubit state $|\phi_x\rangle$. In a seminal paper, Buhrman et al.~\cite{buhrman2001quantum} introduced the notion of {\em quantum fingerprinting}. In their most general form, a quantum fingerprinting is the following.
\begin{definition}[Generalized Quantum Fingerprinting (Buhrman et al.~\cite{buhrman2001quantum})]\label{def:qfinger}
A function $\phi:\{0, 1\}^n\to{\mathbb{C}^{2^m}}$ is a $(n, m, \delta)$ (generalized) quantum fingerprinting where
 $\delta:=\max_{x,x': x\ne x'} \left|\langle\phi_x | \phi_{x'}\rangle\right|$.
\end{definition}

We use the convention that $\phi:=|\phi\rangle\langle\phi|$ represent the projector for the pure state $|\phi\rangle$.
If one replaces the predicate $h(x)=h(x')$ by the fidelity $F(\phi_x, \phi_{x'}) = \left|\langle\phi_x|\phi_{x'}\rangle\right|$, one sees that $\delta$ precisely quantify the extent of collision resistance. For concreteness, we define what we mean by quantum cryptographic hash function as follows. For a function $\delta_n\in(0, 1)$, we say $\delta_n$
is negligible in $n$ if $\delta_n\le 1/n^c$ for all $c>0$ and all sufficiently large $n$.

\begin{definition}\label{def:qhash}
A $(n, m, \delta)$ quantum fingerprinting $\phi$ is a {\em quantum cryptographic hash function}  if $\delta=\negl(n)$. 
\end{definition}

We note that while classical cryptographic hash functions necessarily rely on computational assumptions for security, their quantum counterparts can achieve the three security properties (1-3) information-theoretically.
We now proceed to formulate our leakage problem precisely. We consider average case
security and model classical side-channel information using a classical-classical (c-c) state,
called the {\em side information state},
$
\eta_{XY}:=\sum_{x,y} p_{x, y} |x\rangle\langle x|\otimes|y\rangle\langle y|$
on $\{0, 1\}^n\times\{0, 1\}^{n'}$. Here $X$ is uniformly distributed and $Y$ represents
the side information. The largest probability of correctly guessing $X$ conditioned on $Y$
is $p_g:=p_g(X|Y)_\eta:=\sum_{y} \max_x p_{x, y}$. The {\em conditional min-entropy} is
$H_{\text{min}}(X|Y)_\eta:=-\log p_g(X|Y)_\eta$. We quantify the amount of leakage by 
$k:=n-H_{\text{min}}(X|Y)_\eta$.

The adversary is given the $Y$ sub-system and creates a classical-quantum (cq) state $\rho_{XE}$,
called the {\em forgery state},
through local quantum operations on $Y$. 
The {\em verification scheme} for $\phi$ is the following measurement on $XE$:
$
V := \sum_{x}|x\rangle\langle x|\otimes\phi_x$.
The probability of the forgery state to pass the verification scheme is 
$e_s:=\Tr(V\rho)$. Given the leakage $\ell$, the optimal passing probability of a forgery state is
denoted by $e_s^*:=e_s^*(\ell)$.
We can now define security precisely.

\begin{definition}[Resilience against classical leakage]
A $(n, m, \delta)$ quantum cryptographic hash function $\phi$ is said to be {\em $\sigma$-resilient against $\ell$
bits of classical leakage} if for all forgery state $\rho$ obtained from $\ell$ bits of side information,
the probability of passing the verification scheme $e_s^*(\ell)\leq \sigma$. If no $\sigma$ is specified, it is assumed that $\sigma=\negl(n)$.
\end{definition}

\subsection{Main Result}
We show that quantum cryptographic hash functions can be extremely resilient to classical leakage. Our main theorem is informally stated below.

\begin{theorem}[Main Theorem]\label{thm:main}
For all $n$ and $k=n-\omega(\log n)$, 
all quantum cryptographic hash functions are resilient against $k$ bits of classical leakage.
\end{theorem}

Buhrman et al.~\cite{buhrman2001quantum} showed that for all $n$ and $\delta\in(0, 1)$, there exists
a $(n, m, \delta)$ quantum fingerprinting for $m=\log n+ O(\log1/\delta)$ for which explicit constructions
can be derived from~\cite{naor1993small}. We thus have the following corollary.
\begin{corollary}
For all $n$, $k=n-\omega(\log n)$, and $m=\omega(\log n)$, there exist efficient quantum cryptographic hash functions
resilient to leaking $k$ bits of information. 
\end{corollary}

One drawback of the verification scheme is that the verifier has to get access to full information about the original message $X$ in order to perform the verification. In some cases this would be a heavy burden on the verifier. One natural question to ask is that if it is possible to develop a lightweighted verification scheme where the verifier does not need to read the whole message. More formally, let the verifier now receive $k$ qubits of {\em advice state} and $m$ bits of the {\em
    forgery state} provided
by the adversary. An $(n,k,m)$ verification scheme $V$ would then be a joint measurement on the advice state together with the forgery state. This generalizes the original verification where $k=n$ and $V=\sum_{x}|x\rangle\langle x|\otimes\phi_x$.

Out next result shows that, by increasing the hash a little bit we can dramatically reduce the size of system needed by the verifier:

\begin{theorem}
    For all $n$, fix $k=m=\omega(\log^2 n), \ell\leq n-\omega(\log n)$. There exists a verification scheme $V$ acting on $k+m$ qubits, together with an ensemble of $k$-qubit states $\{\rho_x\}$ such that 
    $$\sup_{Y:H_{\min}(X|Y)\leq \ell}\sup_{\{\sigma_Y\}}\mathbb{E}_{XY}\Tr[V(\rho_X\otimes \sigma_Y)]\leq \negl(n).$$
\end{theorem}

Arunachalam et al.~\cite{arunachalam2016optimal} showed that $\Omega(n)$ copies of a quantum cryptographic hash based on linear codes is necessary to recover the original $n$-bit classical message, regardless of the length of the hash itself. Our result shows the complimentary aspect that only $\omega(\log n)$ copies are sufficient to ensure that the prover holds the classical message.

\commentout{Note here that we can hash a message into near logarithmic number of qubits, with resiliency to any amount of classical information leakage of interest. Specifically we have the following separation lemma:
}
Our central technical result is the following. Recall that $p_g:=2^{\ell-n}$ is the optimal guessing probability of the message conditioned on the $\ell$-bit side information.

\begin{lemma}[Lemma~\ref{lem:sepr}, informally stated]
    \label{lem:sep} For any $(n, m, \delta)$ quantum fingerprinting $\phi$ and any leakage of $\ell$ classical bits, the probability of the forgery state passing the verification scheme satisfies   $$e_s\leq p_g +\delta.$$
\end{lemma}

This implies that
\begin{align}\label{eqn:equiv}
O(p_g + \delta^2) \le e^*_s\leq p_g + \delta,
\end{align}
by considering the two cheating strategies of guessing $X$ first (then applying $\phi$) and using a fixed fingerprint state. Consequently, when $\delta=\negl(n)$, the above inequality means that $e^*_s$ is negligible if and only if $p_g$ is negligible. 

Since $\ell=n-O(\log n)$ is the threshold for $p_g(\ell)$ to be non-negligible, the bounds~(\ref{eqn:equiv}) show for  quantum cryptographic hash functions, the leakage resilience can approach the {\em maximum} of  $n-O(\log n)$ bits. 

One counterpart of this result is shown in~\cite{arunachalam2016optimal}, saying that $\Omega(n)$ copies of quantum fingerprints would be necessary for an adversary to recover the original message with non-negligible probability. Thus, the quantum cryptographic hash functions based on fingerprinting have the following property: The hash itself is efficiently computable, but it is information-theoretically resilient to recovery of the message from the hash (which requires
$\Omega(n)$
copies of the hash) and to recovery of the hash from partial information of the message.

\subsection{Implications on quantum-proof randomness extraction}
Our result reveals some stark contrast between quantum and classical side information. This difference shows the difficulty for establishing the quantum security of classical-proof extractors. Roughly speaking, a randomness extractor is a randomized algorithm which turns a weakly random source into near uniform~\cite{nisan1992pseudorandom,nisan1996randomness,trevisan2001extractors}. These are fundamental objects with a wide range of applications in computational complexity, cryptography, and other areas~\cite{bennett1985reduce,dodis2009non,goldreich1997another,sudan1999pseudorandom,moshkovitz2014parallel}.
In particular, they accomplish the important tasks of privacy amplication~\cite{bennett1995generalized,bennett1985reduce,maurer1993secret}, by decoupling the correlation between the output and the side information.

A major open problem in randomness extraction is whether every classical-proof randomness extractor for $k$ min-entropy sources is secure against quantum adversaries with comparable amount but quantum side information. Loss of parameter is already shown to be inevitable in~\cite{gavinsky2007exponential}, but possibilities still remain in the case where the ranges of parameters are relavant to most typical applications. 

If all quantum side information can be constructed from a comparable amount of classical side information, we would have resolved this major problem positively. Our result shows that this approach would necessarily fail.
For details, see Section~\ref{sec:re}.

\begin{theorem}[Theorem~\ref{thm:topsep}, informally stated]
    There exists a family of classical-quantum states with arbitrarily small amount of quantum side information, yet
    these quantum states cannot be approximately constructed from classical side information that is almost a perfect copy of the classical message.
\end{theorem}

\subsection{Sketch of Proofs}
To prove Separation Lemma~\ref{lem:sep},  we first answered the following question (see Section~\ref{sec:cgp}): 
{\em Given a c-q state, how much classical side information is needed to approximate it up to a given error $\epsilon$?}
It turns out that this quantity, denoted by the \emph{conversion parameter}, can be simplified to a fairly nice expression. An important step is to show (in Lemma~\ref{lem:ccsimp}) that  classical information can without loss of generality be identifying a subset of the messages, and conditioned on this information, the messages are uniformly distributed on that subset. 
At the end, the proof reduces to estimating the operator norm of  $\sum_iT_i$, where $T_i$'s are projections to the fingerprint states, thus pairwise almost orthogonal. Cotlar-Stein Lemma gives us the desired bound. 

\subsection{Related works.}
As mentioned above, Buhrman et al.~\cite{buhrman2001quantum} introduced the notion and provide the constructions of quantum fingerprinting. The application they focused on is for message identification. For our cryptographic applications, we are primarily interested in instances of a negligible fidelity. They did not discuss properties of quantum fingerprinting in an adversarial context like ours. That quantum fingerprinting satisfies the security properties of cryptographic hash functions was observed and explored by~\cite{ablayev2013cryptographic,vasiliev2014cryptographic}. 

Side-channel attack is a major paradigm studied in the classical information security and cryptography community due to its high level of threat in practice~\cite{side-key,naor2009public,akavia2009simultaneous,dodis2010efficient,brier2002weierstrass,joye2001hessian}. Side-channel key recovery attack has in particular drawn much attention~\cite{side-key}. However, these classical works address problems that necessarily require computational assumptions and many other works focus on the hardware aspects. To the best of our knowledge, this work appears to be the first studying information theoretical security of quantum cryptography against classical side-channel attack.

%% file: HuangS16_QIP17.technical_content.tex

\newcommand{\X}{\mathcal{X}}
\newcommand{\Y}{\mathcal{Y}}
\newcommand{\St}{\mathcal{S}}
\newcommand{\I}{\mathcal{I}}
\newcommand{\hro}{\hat{\rho}}
\newcommand{\Supp}{\mathrm{Supp}}
\newtheorem{conj}{Conjecture}

\section{Preliminaries}
In the subsequent subsections, we summarize the necessary background on quantum information, probability theory and operator theory. We assume that readers are familiar with linear algebra. We refer to~\cite{nielsen2010quantum} for more detailed knowledge on quantum information and quantum computation.
\subsection{Quantum Information}
\paragraph{Quantum States.} The state space of a $k$-qubit system can be characterized as a $2^k$-dimensional Hilbert space $\mathcal{H}\cong\mathbb{C}^{2^k}$. A $k$-qubit quantum state is described by a density operator $\rho$ on $\mathcal{H}$, i.e.\ a positive semidefinite operator with unit trace. We often use $\rho\succcurlyeq 0$ to denote that $\rho$ is positive semidefinite. The set of all density matrices on $\mathcal{H}$ is then denoted by $\mathcal{S}_{=}(\mathcal{H}):=\{\rho\succcurlyeq 0|\Tr[\rho]=1\}$. Sometimes we are
        also interested in the set of subnormalized states $\mathcal{S}_{\leq }(\mathcal{H}):=\{\rho\succcurlyeq 0|\Tr[\rho]\leq 1\}$.

        Let $L(\mathcal{H})$ denote the set of all linear operators on $\mathcal{H}$ onto itself. For two operators $A,B\in L(\mathcal{H})$, one can define the inner product $\langle A,B\rangle$ to be $\Tr[A^\dag B]$, $A^\dag$ being the adjoint conjugate of $A$. Also we denote the identity operator by $I$.

        The state space for a multipartite system is the tensor of the state space of each individual system, i.e.\ a general multipartite state $\rho_{ABC}$ can be described by a density operator in $\mathcal{S}_{=}(\mathcal{A}\otimes \mathcal{B}\otimes \mathcal{C})$. For such a state $\rho$, the marginal state on a subsystem, say $\rho_A$, is the partial trace over the other subsystems, namely $\rho_A=\Tr_{BC}[\rho_{ABC}]$.

        A \emph{classical-quantum-}, or cq-state is a state of the form
        $$\rho_{XE}=\sum_{x}|x\rangle\langle x|\otimes \rho_x,$$
        where $\{|x\rangle\}$ are orthonormal and $\rho_x$'s are subnormalized states which sum up to a density matrix. This is to say, the bipartite state $\rho_{XE}$ is classical on the $X$ side while quantum on the $E$ side. Furthermore, if the $E$ side is also classical, then the state is of the form
        $$\rho_{XE}=\sum_{x,y}p(x,y)|x\rangle\langle x|\otimes |y\rangle\langle y|.$$
        Such a state is called a \emph{classical-classical-} or cc-state. Note that cc-states can be identified as classical joint distributions.
    \paragraph{Quantum Measurement.}
        The most general type of quantum measurement is \emph{Positive-operator-valued measurement (POVM)}. A POVM is a set of measurements $\mathcal{M}=\{M_o\}_{o\in O}$, where $O$ is the set of all outcomes, $M_o$ is positive semidefinite for all $o$ and $\sum_{o\in O}M_o=I$. For a state $\rho$, the probability that the outcome is $o$ under the measurement $\mathcal{M}$ is $\Tr[M_o\rho]$.
    \paragraph{Quantum Channel.}
        A physically realizable quantum channel $\mathcal{C}_{A\rightarrow B}$ is a completely positive, trace preserving linear map from $L(\mathcal{A})$ to $L(\mathcal{B})$.

        Sometimes we may consider more restrictive quantum channels, such as classical-to-quantum channels and classical-to-classical channels (note that quantum-to-classical channels can be identified as measurements). A classical-to-quantum channel $\mathcal{C}_{A\rightarrow B}$ can be written as
        $$\mathcal{C}(\cdot)=\sum_{a}\langle a| \cdot|a\rangle \cdot\rho_a,$$
        where $\rho_a\in\mathcal{S}_{=}(\mathcal{B})$. When $B$ is also classical, the channel conincides with the standard definition of a channel in classical information theory, specified by the conditional distribution $p_{B|A}$.
\subsection{Probability Theory}
    \paragraph{Random Source.} A distribution $X$ over a finite set $\X$ is also called a random source in this work. Specifically, denote $U_{\X}$ the uniform distribution over the set $\X$. For a random source $X$, we can define \emph{support set} 
        $$\Supp(X):=\{x\in\X|\Pr[X=x]>0\},$$
        and \emph{guessing probability} 
        $$p_g(X)=\max_{x}\Pr[X=x].$$
        The operational definition of guessing probability is as follows: if one is asked to guess the value generated by the random source $X$ without any side information, the optimal strategy is to guess the most probable one, and the winning probability would be $p_g(X)$.
    \paragraph{Statistical Distance.} For two distributions $X_1, X_2$ over a finite set $\X$, the statistical distance is defined to be
        $$\Delta(X_1,X_2):=\frac{1}{2}\sum_x|\Pr[X_1=x]-\Pr[X_2=x]|=\max_{S\subseteq \X}\Pr[X_1\in S]-\Pr[X_2\in S].$$
        The operational definition of statistical distance is as follows: if one is asked to tell whether a given message $x$ is generated from source $X_1$ or $X_2$, the success probability under optimal strategy is $\frac{1+\Delta(X_1,X_2)}{2}$.
    \paragraph{Random Source with Classical Side Information.} For a joint distribution $XY$ over a set $\X\otimes \Y$, the \emph{guessing probability} of $X$ conditioned on $Y$ is given by
        $$p_g(X|Y):=\mathbb{E}_{y\sim Y}p_g(X|Y=y)=\sum_{y}\max_x \Pr[X=x,Y=y].$$
        This is to say, upon possessing classical side information $y$ of the random source $X$, the optimal strategy to guess the value of $X$ is simply guessing the most probable one conditioned on $Y=y$. The overall winning probability would then be $p_g(X|Y)$.
    \paragraph{Quantum Random Source.} Now suppose that the side information can be quantum, i.e.\ the random source $X$ together with the quantum side information $\rho_E$ form a classical-quantum (cq) state
        $$\rho_{XE}=\sum_{x}|x\rangle\langle x|\otimes \rho_x,$$
        where the $X$ part is still classical while the $E$ part can be quantum. Note that $\rho_x$'s are subnormalized and they sum up to a normalized state $\rho_E$. Then the guessing probability of $X$ conditioned on $E$ is defined as the guessing probability under an optimal POVM:
        $$p_g(X|E)_\rho := \max_{\mathcal{M}_{E\rightarrow Y}}p_g(X|Y)_{(\mathcal{I}\otimes\mathcal{M)(\rho)}}.$$
        It can be shown using SDP duality that 
        $$p_g(X|E)_\rho = \min_{\sigma: \forall x, \sigma\succcurlyeq \rho_x}\Tr[\sigma].$$

        When the number of qubits in the quantum side information is a small number $k$, we have the following bound for the guessing probability:
        $$p_g(X|E)\leq p_g(X)\cdot 2^{k}.$$
\subsection{Operator Theory}
For an operator $K\in L(\mathcal{A})$, define the $p$-Schatten norm to be
$$\Vert K\Vert_p = \Tr\left[(K^\dag K)^{p/2}\right]^{1/p},$$
for all $1\leq p<\infty$. Some concrete examples are as follows:
\begin{itemize}
    \item $p=1$, then $$\Vert K\Vert_{tr} = \Vert K\Vert_1=\Tr[\sqrt{K^\dag K}]$$ is called the \emph{trace norm} of $K$. For two quantum states $\rho,\sigma\in \mathcal{S}_{=}(\mathcal{A})$, the \emph{trace distance} $\Vert \rho-\sigma\Vert_{tr}$ is the maximum $l_1$-distance of distributions obtained via any quantum measurement.
    \item $p=\infty$, then the $\infty$-Schatten norm is defined to be
        $$\Vert K\Vert_{\infty} = \lim_{p\rightarrow }\Vert K\Vert_p.$$
        This is also called the \emph{operator norm}
        $$\Vert K\Vert_{op}:=\sup_{v\neq 0}\Vert Kv\Vert/\Vert v\Vert.$$
        In the case that $K$ is positive semidefinite, $\Vert K\Vert_{\infty}$ is the largest eigenvalue of $K$.
    \end{itemize}

\section{Characterizing the Smooth Conversion Parameter}
\label{sec:cgp}
Given a cq state $\rho_{XE}$, we are now interested in determining the least amount of classical correlation,  measured by the \emph{conversion parameter}, that one needs in order to generate $\rho$ by applying a quantum channel on the side information.
\begin{definition}[Conversion Parameter]
    \label{defn:cgp}
    Let $\rho_{XE}\in \St_{\leq}(\X\otimes \E)$ be a cq state. The \emph{conversion parameter} of $X$ conditioned on $E$ is defined as
    $$p_{\downarrow}(X|E)_{\rho}:=\min_{\substack{\eta_{XY},\C\\ \I\otimes \C(\eta)=\rho}}p_{g}(X|Y)_\eta.$$
\end{definition}

By monotonicity of guessing probability under quantum channels acting on the side information, we have $p_{\downarrow}(X|E)_{\rho}\geq p_{g}(X|E)_{\rho}$. 

The definitions above would be less interesting if we do not allow the existence of an extra error term $\epsilon$ for the following reason. For a state $\rho$ with the form
$$\rho_{XE}=\sum_{x}q_x|x\rangle\langle x|\otimes |\psi_x\rangle\langle \psi_x|,$$
where $|\psi_x\rangle$ for each $x$ are distinct, the only classical side information to generate $\rho$ losslessly would be the classical message itself; but if we allow that the generated state be $\epsilon$-close to the desired state $\rho$, we may be able to approximate $\rho$ using classical information with significantly lower guessing probability. 

\begin{definition}[Smooth Conversion Parameter]
    \label{defn:scgp}
    Let $\epsilon\geq 0$ and $\rho_{XE}\in \St_{=}(\X\otimes \E)$. Then the \emph{$\epsilon$-smooth conversion parameter} of $X$ conditioned on $E$ is defined as
    $$p^{\epsilon}_{\downarrow}(X|E)_{\rho}:=\min_{\substack{\eta_{XY},\C\\ \Vert \I\otimes \C(\eta)-\rho\Vert_{tr}\leq \epsilon}}p_{g}(X|Y)_\eta.$$
\end{definition}

One may wonder if we can derive a good estimate of $p^{\epsilon}_{\downarrow}(X|E)_\rho$ for given parameter $\epsilon$. Recall that for guessing probability $p_g(X|E)_{\rho}$, we have
$$p_{g}(X|E)_{\rho}=\min_{\sigma:\forall x,\sigma\succcurlyeq \rho_x}\Tr[\sigma]$$
by SDP duality of quantum guessing probability. Similarly, we can prove the following lemma for smooth conversion parameter:
\begin{lemma}
    \label{lem:cgp}
    Let $\epsilon\geq0$ and $\rho_{XE}=\sum_x|x\rangle\langle x|\otimes \rho_x\in \St_{=}(\X\otimes \E)$, then
    $$p^{\epsilon}_{\downarrow}(X|E)_{\rho}=\min_{\substack{p_S,\hro_S\\ \sum_x\Vert\sum_{x\in S}p_S\hro_S-\rho_x\Vert_{tr}\leq \epsilon}}\sum_Sp_S,$$
    where $S$ ranges over all nonempty subsets of $\X$.
\end{lemma}

Before proceeding to the proof of the lemma, we want to first narrow down the set of classical side information we need to consider.

\begin{lemma}
    \label{lem:ccsimp}
    Let $\eta_{XY}=\sum_{xy}p_{xy}|x\rangle\langle x|\otimes |y\rangle\langle y|$ be a cc state. Then there exists a cc state $\eta'_{XS}\in \St_=(\X\otimes 2^{\X})$ with the following three properties:
    \begin{enumerate}
        \item $p_g(X|Y)_\eta = p_g(X|S)_{\eta'}$;
        \item There exists a classical channel $C_{S\rightarrow Y}$ such that $(I\otimes C)(\eta')=\eta$;
        \item $\eta'$ is of the form
            $$\eta'=\sum_{S\subseteq \X}p_{S}(\sum_{x\in S}|x\rangle\langle x|)\otimes |S\rangle\langle S|.$$
    \end{enumerate}
\end{lemma}
This lemma implies that, for our purposes, the set of all joint distributions we need to consider is of the form where the side information only indicates in which set the classical information is uniformly distributed.
\begin{proof}
    We provide a constructive proof. For sake of simplicity we will use classical probability notations, i.e.\ $\eta_{XY}$ will be identified as the classical joint distribution $p_{XY}$.

    All subset of $\X$ form a directed acyclic graph under containment relations, therefore we can recursively define 
    $$p_S(y)=\min_{x\in S}p(x,y)-\sum_{S\subsetneq S'}p_{S'}(y)$$
    for all nonempty $S\subseteq\X$. Let $\hat{p}_S=\sum_yp_S(y)$ and $C_S=\Supp(p_S(y))$. We claim that the distribution 
    $$p'(x,S)=\hat{p}_S\delta_{x\in S},$$
    together with the channel
    $$C(y|S)=p_S(y)/\hat{p}_S$$
    satisfies our requirements for $\eta'$. 

    First we need to show that $p'_{XS}$ is indeed a distribution, and $C$ is indeed a classical channel. It suffices to show that $p_S(y)\geq0$ for all $S$ and $y$. 
    
    To prove this, we apply induction on size of the set to prove two things: 
    \begin{itemize}
    \item $p_S(y)\geq 0$ for all $S$ and $y$;
    \item For all $S_1$ and $S_2$ where no one contains the other, we have $C_{S_1}\cap C_{S_2}=\emptyset$.
    \end{itemize}
    The base case is when $|S|,|S_1|,|S_2|\geq|\X|$. Apparently we have $S=S_1=S_2=\X$, so the second statement is true. For the first one, we have by construction
    $$p_{\X}(y)=\min_xp(x,y)\geq0.$$
    Let's move on to the induction step. Assume now that for all $S_1,S_2$ where\begin{itemize}\item $|S_1|,|S_2|> k$\item $S_1,S_2\subsetneq S_1\cup S_2$\end{itemize}
    we have $p_S(y)\geq 0$ for all $y$ and $C_{S_1}\cap C_{S_2}=\emptyset$. Then for a given $y$ and a subset $S$ with $|S|=k$, the collection of subsets
    $$\{S'|y\in \Supp(S'),S\subsetneq S'\}$$
    must form a chain $S'_1\subsetneq S'_2\subsetneq \cdots\subsetneq S'_m$ under containment relations. Then
    \begin{align*}
        p_S(y)&=\min_{x\in S}p(x,y)-\sum_{S\subsetneq S'}p_{S'}(y)\\
        &=\min_{x\in S}p(x,y)-\sum_{i=1}^m p_{S'_i}(y)\\
        &=\min_{x\in S}p(x,y)-\left(\min_{x\in S'_1}p(x,y)-\sum_{i=2}^m p_{S'_i}(y)\right)-\sum_{i=2}^m p_{S'_i}(y)\\
        &=\min_{x\in S}p(x,y)-\min_{x\in S'_1}p(x,y)\\
        &\geq0.
    \end{align*}
    Also, for $S_1,S_2$ not contained in each other and with size $\geq k$, we have
    \begin{align*}
        p_{S_1}(y)&=\min_{x\in S_1}p(x,y)-\sum_{S_1\subsetneq S'}p_{S'}(x,y)\\
        &=\left(\min_{x\in S_1}p(x,y)-\min_{x\in S_1\cup S_2}p(x,y)\right)-\sum_{S_1\subseteq S', S_2\not\subseteq S'}p_{S'}(y)\\
        &\leq \min_{x\in S_1}p(x,y)-\min_{x\in S_1\cup S_2}p(x,y)\end{align*}
     and similarly 
     $$p_{S_2}(y)\leq\min_{x\in S_2}p(x,y)-\min_{x\in S_1\cup S_2}p(x,y)$$
     using the induction hypothesis $p_S(y)\geq0$ for all $|S|\geq k$.
    As $$\min_{x\in S_1\cup S_2}p(x,y)=\min\{\min_{x\in S_1}p(x,y),\min_{x\in S_2}p(x,y)\},$$ we know that either $p_{S_1}(y)= 0$ or $p_{S_2}(y)= 0$, which results in that $y\notin C_{S_1}\cap C_{S_2}$. As this holds for an arbitrary $y$, we must have $C_{S_1}\cap C_{S_2}=\emptyset$. This completes the induction step.
    
	We then proceed to prove that $p'$ and $C$ satisfies the three properties listed in Lemma~\ref{lem:ccsimp}. Note that property 3 is automatically satisfied by the form of $p'$, so it suffices to show just the first two properties.
    
    \begin{enumerate}
    \item To see property $1$ of Lemma~\ref{lem:ccsimp}, note that the guessing probability of $p$ and $p'$ are respectively
    $p_g(X|Y)_{p}=\sum_y\max_xp(x,y)$
    and $p_g(X|S)_{p'}=\sum_S\hat{p}_S=\sum_y\sum_Sp_S(y)$. It then suffices to prove that $\sum_S p_S(y)=\max_x p(x,y)$ holds for all $y$. 
    
    One direction is easier to prove: for all $x$ we have
    \begin{align*}
        \sum_{S}p_{S}(y)&\geq \sum_{\{x\}\subseteq S}p_{S}(y)\\
        &=p(x,y)-\sum_{\{x\}\subsetneq S}p_S(y)+\sum_{\{x\}\subsetneq S}p_S(y)\\
        &= p(x,y),
    \end{align*}
    thus $\sum_{S}p_S(y)\geq \max_x p(x,y)$. To prove the other direction, let's look more in detail into what we have proved so far. For every subset $S$, we have shown that the sum $\sum_{S\subsetneq S'}p_{S'}(y)$ can be reduced to summation on a chain. Without loss of generality, this chain can be completed to maximal length such that the difference of cardinalities of adjacent terms on this chain is $1$, as additional terms would not change the final result. As
    $$\sum_{S}p_S(y)=\sum_{\emptyset\subsetneq S}p_S(y),$$ there exists a chain $\emptyset\subsetneq S_1\subsetneq\cdots\subsetneq S_{|\X|}=\X$ such that $|S_i|=i$ and
    $$\sum_{S}p_S(y)=\sum_{i=1}^{|\X|}p_{S_i}(y).$$ 
    As $|S_1|=1$, there must exists $x^*$ such that $S_1=\{x^*\}$. Then $$\sum_Sp_S(y)=\sum_{i=1}^{|\X|}p_{S_i}(y)=p(x^*,y)\leq \max_x p(x,y),$$ which proves the opposite direction of property 1.
	\item  For property 2, denote $p''_{XY}$ the joint distribution we get from applying $C$ onto $p'$. We have
    \begin{align*}
        p''(x,y)&=\sum_{S}p'(x,S)C(y|S)\\
        &=\sum_S \delta_{x\in S} \hat{p}_S\cdot p_S(y)/\hat{p}_S\\
        &=\sum_{x\in S}p_S(y)\\
        &=p_{\{x\}}(y)+\sum_{\{x\}\subsetneq S}p_S(y)\\
        &=p(x,y)-\sum_{\{x\}\subsetneq S}p_S(y)+\sum_{\{x\}\subsetneq S}p_S(y)\\
        &= p(x,y).
    \end{align*}
    This proves property 2, which finishes the entire proof of Lemma~\ref{lem:ccsimp}.
    \end{enumerate}
\end{proof}
With Lemma~\ref{lem:ccsimp}, we are now ready to prove Lemma~\ref{lem:cgp}.
\begin{proof}[Proof of Lemma~\ref{lem:cgp}]
	Recall the definition of $\epsilon$-smooth conversion parameter:
    $$p^{\epsilon}_\downarrow(X|E)_\rho=\min_{\substack{\eta_{XY},\C_{\Y\rightarrow \E}\\\Vert(\I\otimes \C)(\eta)-\rho\Vert\leq \epsilon}}p_g(X|Y)_\eta.$$
    For every $\eta$, by Lemma~\ref{lem:ccsimp} there exists $\eta'_{XS}=\sum_Sp_S(\sum_{x\in S}|x\rangle\langle x|)\otimes |S\rangle\langle S|$ satisfying the three properties. It is easy to see that $\eta'$, with the channel $\C'=\C\circ C$, is also a cc state achieving the same guessing probability and error, and the guessing probability of $\eta'$ is $\sum_Sp_S$. A general quantum channel acting on the side information can be characterized as follows without loss of generality:
    $$\C'_{2^{\X}\rightarrow \E}(\cdot)=\sum_{S}\langle S|\cdot|S\rangle\hro_S.$$
    Then we have
    $$(\I\otimes \C')(\eta')=\sum_S(\sum_{x\in S}|x\rangle\langle x|)\otimes p_S\hro_S,$$
    which leads to
    $$\Vert (\I\otimes \C')(\eta')-\rho\Vert=\sum_x\Vert\sum_{x\in S}p_S\hro_S-\rho_x\Vert,$$
    finishing the proof of Lemma~\ref{lem:cgp}.
\end{proof}

\section{Proof of Separation Lemma}
\label{sec:prf}
In this section we provide a proof for our main technical result, the Separation Lemma~\ref{lem:sepr}.

\begin{lemma}[Separation Lemma]
\label{lem:sepr}
For all $\delta>0, \epsilon\geq 0$ and state
$$\rho_{XE}=\sum q_x|x\rangle\langle x|\otimes |\psi_x\rangle\langle \psi_x|$$
with \emph{maximum overlap} 
$$\delta=\max_{x,y:x\neq y}|\langle\psi_x| \psi_y\rangle|,$$

we have
    $$1-\epsilon/2\leq (1-\delta)p_{\downarrow}^{\epsilon}(X|E)_\rho+\delta.$$
\end{lemma}

Regarding the quantum side information as a cryptographic hash function, now let us consider the following scheme. An adversary, upon receiving some classical leakage of a classical message, tries to forge a quantum system which looks like the authentic hash. The adversary then sends this state to a verifier with full knowledge of the classical message, and the verifier, conditioned on classical information $x$, performs a projection measurement $V$ on the quantum state provided by the
adversary. We say that the adversary cheats the verifier if the measurement outcome is accept. Given that the forged state is $\epsilon$-close to the quantum hash, the passing probability $e_s$ is clearly no less than $1-\epsilon/2$, and the probability $p_g$ that the adversary can guess the classical message $X$ correctly is no more than $p^{\epsilon}_{\downarrow }(X|E)_{\rho}$. It then suffices to prove that in this particular scheme, we have
$$e_s\leq (1-\delta)p_g+\delta.$$

In the proof of Lemma~\ref{lem:sepr} we will make use of Cotlar-Stein Lemma, which gives a good estimate of the operator norm of the sum of near-orthogonal operators. We will attach the proof of the Cotlar-Stein Lemma in Appendix~\ref{sec:csl} for completeness.

\begin{lemma}[Cotlar-Stein Lemma]
\label{lem:csl}
For a set of unit vectors $\{|\psi_1\rangle,|\psi_2\rangle,\cdots,|\psi_n\rangle\}$ with maximum fidelity $\max_{i,j:i\neq j}|\langle \psi_i|\psi_j\rangle|\leq \delta$, we have
$$\lambda_{\max}\left(\sum_{i=1}^n|\psi_i\rangle\langle \psi_i|\right)\leq 1+(n-1)\delta.$$
\end{lemma}
\begin{proof}[Proof of Lemma~\ref{lem:sepr}]
By Lemma~\ref{lem:cgp}, without loss of generality, we can assume that the joint state of the classical message and the side information obtained by the adversary is of the form
$$\rho_{XZ}=\sum_{S}p_S(\sum_{x\in S}|x\rangle\langle x|)\otimes |S\rangle\langle S|$$
where $\sum_{x\in S}p_S=q_x$ for all $x$. It follows then that $$p_g=\sum_Sp_S, \sum_{S}p_S|S|=1.$$ The adversary then prepares state $\hro_S$ based on the side information $S$. We have a precise evaluation of $e_s$:
\begin{align*}
e_s&=\sum_x \langle \psi_x|\sum_{x\in S}p_S\hro_S|\psi_x\rangle \\
&=\sum_S p_S\Tr[\hro_S \sum_{x\in S}|\psi_x\rangle\langle \psi_x|]\\
&\leq \sum_S p_S\lambda_{\max} (\sum_{x\in S}|\psi_x\rangle\langle \psi_x|).\\
\end{align*}
By the Cotlar-Stein Lemma, we have
$$\lambda_{\max}(\sum_{x\in S}|\psi_x\rangle\langle \psi_x|)\leq 1+(|S|-1)\delta.$$
Since we have $p_g=\sum_Sp_S$ and $1=\sum_Sp_S|S|$, 
\begin{align*}
e_s&\leq \sum_Sp_S\lambda_{\max}\left(\sum_{x\in S}|\psi_x\rangle\langle\psi_x|\right)\\
&\leq \sum_Sp_S\left(1+(|S|-1)\delta\right)\\
&=p_g(1-\delta)+\delta.\end{align*}
\end{proof}

\section{Quantum Hash}
\subsection{Quantum Fingetprinting is Resilient to Classical Leakage}
Now let us consider the setting of quantum hashing. For the sake of simplicity, we assume that the classical message $X$ is uniformly distributed over $\{0,1\}^n$, and all parties are information theoretic and considered as channels.

Suppose now that a prover $A$, either adversarial or honest, upon obtaining some side information $Y$ of $X$, wants to show that he has full access to the classical message $X$. To show this, he needs to pass a test held by a verifier $V$ who has full access to $X$. One way to do this is to send $V$ an $m$-bit hash $M$, and $V$ will accept or reject based on the classical messages $X$ and $M$.

If the adversary wants to cheat as best as he can, the optimal strategy would be applying a deterministic mapping on the side information $Y$; similarly, if the verifier wants to distinguish adversarial parties from honest ones as best as he can, the optimal strategy would also be a deterministic algorithm. Therefore we can safely assume that both $A$ and $V$ are deterministic mappings.

\begin{definition}[Resilience against leakage]
    A mapping $h:\{0,1\}^n\rightarrow\{0,1\}^m$ is called $\sigma$-resilient against $k$ bits of information leakage if for all classical side information $Y$ of $X$ such that $|Y|\leq k$, forall mappings $f:\{0,1\}^k\rightarrow \{0,1\}^m$, we have
        $$\Pr_{XY}[f(Y)=h(X)]\leq \sigma.$$
        If $\sigma$ is not specified then it is assumed that $\sigma=\negl(n)$.
\end{definition}

Ideally, one would want a hash function which is resilient against much information leakage and short at the same time.
Unfortunately these two requirements cannot be achieved at the same time. If an $m$-bit verification scheme $V$ is resilient against $k$ bits of classical leakage, one must have $m=k+\omega(\log n)$. This can be seen as follows: suppose otherwise that $m=k+O(\log n)$. Then an
adversarial prover, upon getting the first $k$ bits of the message an honest party would send to the verifier, guesses the remaining $O(\log n)$ bits uniformly at random. Such a prover would then have an inverse polynomial probability of passing the test.

Now suppose that both the verifier and the prover have access to quantum power, while the information leakage is still classical. Now the prover can send an $m$-qubit quantum system $\rho$ to the verifier, and the verifier would then perform a joint measurement on both $X$ and $\rho$ to determine whether to accept or not. Denote the state space of $X,Y$ and $M$ respectively by $\mathcal{X},\mathcal{Y}$ and $\mathcal{M}$.

\begin{definition}[Resilience of quantum fingerprinting against classical information leakage]
    An $(n,m,\delta)$ quantum fingerprinting $\phi$ is called $\sigma$-resilient against $k$ bits of classical information leakage if for all classical side information $Y$ of $X$ such that $|Y|\leq k$ and any quantum channel $\mathcal{C}_{\mathcal{Y}\rightarrow \mathcal{M}}$, we have
            $$\Tr[V\cdot (\mathcal{I}\otimes\mathcal{C})(\rho_{XY})]=\sigma,$$
            where $$\rho_{XY}=\sum_{x,y}\Tr[X=x,Y=y]|x\rangle\langle x|\otimes |y\rangle\langle y|,$$
            and $$V=\sum_{x}|x\rangle\langle x|\otimes \phi_x.$$
            When $\sigma$ is not specified, it is assumed that $\sigma=\negl(n)$.
\end{definition}

In the case where all states we are considering are cq states, it is safe to replace a general channel by mapping each classical information to a state. Therefore,  $\mathcal{C}$ can be specified by
$$\mathcal{C}(\cdot)=\sum_y\langle y|\cdot| y\rangle\rho_y.$$
Then the passing probability $\Tr[V\cdot (\mathcal{I}\otimes\mathcal{C})(\rho_{XY})]$ can be written more elegantly as $\mathbb{E}_{XY}[\Tr[\phi_X\rho_Y]]$.

In sharp contrast to the classical case, the Separation Lemma implies that there exists a $(n,m,\delta)$ quantum fingerprinting resilient against $k$ bits of classical information leakage, where $k$ is much larger than $m$. In fact we have the following theorem.

\begin{theorem}\label{thm:main_exist}
    For all $n$, $m=\omega(\log n)$ and $k=n-\omega(\log n)$, there exists a quantum $(n,m)$ cryptographic hash which is resilient against $k$ bits of classical information leakage.
\end{theorem}

Before proving this theorem, note that this theorem is tight on both sides. If $k=n-O(\log n)$, then an adversary knowing the side information can guess correctly the actual value of $X$ with probability inverse polynomial, thus the success probability would also be non-negligible; on the other hand, if $m=O(\log n)$, we claim that an adversary with zero side information can still pass the test with non-negligible probability. 

To see this, let's start from the completeness condition. This is saying that there exist $\rho_x$'s such that $\mathbb{E}_X[\Tr[M_X\rho_X]]=1$. This can only happen when each term is 1, which in turn implies that $\Tr[M_x]\geq 1$ for all $x$. Now suppose the adversary has no side information about the classical message, so the best he can do is to prepare a state $\rho$. The success probability will then be $$\mathbb{E}_X\Tr[M_X\rho]\leq \lambda_{\max}(\mathbb{E}_X[M_X]),$$
which can be approached when $\rho_0=|\psi_0\rangle\langle \psi_0|$, $|\psi_0\rangle$ being the eigenvector corresponding to the largest eigenvalue of $\mathbb{E}_X[M_X]$. The operator norm can be lower bounded from the trace by 
$$\lambda_{\max}(\mathbb{E}_X[M_X])\geq \frac{\Tr[\mathbb{E}_X[M_X]}{\dim \mathcal{M}}\geq\frac{1}{poly(n)},$$
resulting in a non-negligible passing probability without any side information.

Now let us proceed to the proof of Theorem~\ref{thm:main_exist}.
\begin{proof}[Proof of Theorem~\ref{thm:main_exist}]
    Take an $(n,m,\delta)$ quantum fingerprinting $\phi$, where $\delta$ will be specified later. We know that such a finderprinting exists for $m=O(\log n+2\log\frac1\delta)$, therefore there exists $\delta=\negl(n)$ such that $\phi$ is a quantum cryptographic hash. The verification scheme associated to this quantum fingerprinting is then
    $$V=\sum_{x}|x\rangle\langle x|\otimes \phi_x.$$
    Upon getting $k = n-l$ bits of classical information, the guessing probability of the adversary to the classical message is upper bounded by $2^{-l}$. By Separation Lemma, the probability that the adversary pass the test $e_s$ is upper bounded by $2^{-l}+\delta$, which is still a negligible function of $n$ given $l=\omega(\log n)$.
\end{proof}
\subsection{Generalized Verification Scheme Resilient to Classical Leakage}
Our verification scheme is maximally resilient to classical leakage of information. However, it is still not good enough because the verifier may need to get full access to the whole message. One may ask if it is possible that the verifier only use a small amount of information from the message to perform the verification scheme, yet the verification is still resilient to classical leakage.

To formulate this idea, we generalize the definition of a verification scheme $V=(\mathcal{C},M)$ with
three parameters $(n,k,m)$ played by two players $A$ and
$B$ as follows:
\begin{enumerate}
    \item A joint distribution $XY$, where the marginal distribution of $X\in\{0,1\}^n$ is uniformly random, is generated by nature.
    \item $A$ gets the $k$-qubit quantum state $\rho_X=\mathcal{C}(|X\rangle\langle X|)$ and $B$ gets the side-information string $Y$.
    \item B generates an $m$-qubit quantum state $\mu_Y$ and sends it to $A$.
    \item A perform a joint measurement $M$ on the state $\rho_X\otimes \mu_Y$. The game is successful if and only if the measurement accepts. The overall success probability is thus
        $$e_s^V=\mathbb{E}_{XY}[M(\rho_X\otimes\mu_Y)].$$
\end{enumerate}

For a given $\ell$, define the optimal success probability over all classical leakage $Y$ where $H(X|Y)\geq \ell$ to be $e_s^{V*}(\ell)$.
We can then define the resilience formally:

\begin{definition}
    \label{defn:vsec}
    A $(n,k,m)$-verification scheme $V=(\mathcal{C},M)$ is called $\sigma$-resilient to $\ell$ bits of classical leakage if we have
    $$e_s^{V*}(\ell)\leq \sigma.$$
    In the case where $\sigma$ is not specified, it is assumed that $\sigma$ is negligible.
\end{definition}

Interestingly, one can use the power of quantum side information to reduce the size of the advice state.

\begin{theorem}
    For all $n$, there exists a $(n,k,m)$-verification scheme $V$ which is resilient to $\ell$ bits of classical leakage whenever $$k=m=\omega(\log^2n), n-\ell=\omega(\log n).$$
\end{theorem}

\begin{proof}
    Proof by construction. Take $t,m'=\sqrt{k}=\omega(n)$. We first fix a $(n,m')$-quantum cryptographic hash function $\phi$. for a given message $x$, the advice state would just be $t$-fold tensor product of $\phi_x$'s, namely
    $$\rho_x:=\phi_x^{\otimes t}.$$

    Upon receiving the hash $\mu_Y$ consisting of $t$ parts of $m'$-qubit states, the verifier performs SWAP test between all $t$ pairs of $\phi_X$ and each of the qubit states, accepts if all SWAP tests pass and rejects otherwise. Note that an honest party having full access to $X$ would be able to produce $\phi_X^{\otimes t}$ perfectly, thus successes with probability 1.

    To see that this scheme is resilient against classical leakage, assume now that the state received by the verifier is a forgery state $\mu_Y$. Regarding both $\rho_X$ and $\mu_Y$ as $t$-partite states, we use $\rho_X^T, \mu_Y^T, SWAP^T$ to denote the marginal state on subsystems $T\subseteq[t]$ and the swap between the two subsystems respectively. The measurement corresponding to one copy of SWAP test is $\frac{I+SWAP}{2}$, thus the success probability will be
    \begin{align*}
    e_s^V=&\mathbb{E}_{XY}\left[\frac{\Tr[(I+SWAP)^{\otimes t}(\rho_X\otimes \mu_Y)]}{2^t}\right]\\
    =&\frac{1}{2^t}\sum_{T\in[t]}\mathbb{E}_{XY}\left[\Tr[\bigotimes_{i\in T}SWAP^i(\rho_X\otimes \mu_Y)]\right]\\
    =&\frac{1}{2^t}\sum_{T\in[t]}\mathbb{E}_{XY}\left[\Tr[SWAP^T(\rho_X^T\otimes \mu_Y^T)]\right]\\
    =&\frac{1}{2^t}\sum_{T\in[t]}\mathbb{E}_{XY}[\Tr[\rho_X^T \cdot \mu_Y^T]].\\
\end{align*}
Here the third line comes by tracing out irrelevant states, and the fourth line comes from the identity $\Tr[SWAP(\rho\otimes \sigma)]=\Tr[\rho\sigma]$. For each term $\Tr[\rho_X^T\cdot \mu_Y^T]$ with nonempty $T$ and any $i\in T$, we have
$$\Tr[(\rho_X^T\cdot \mu_Y^T)]\leq \Tr[(\phi^i\otimes I^{T\setminus\{i\}})\mu_Y^T]=\Tr[\phi_x\rho_Y^i].$$
This gives us
$$e_s^V\leq \frac{1}{2^t}+\max_i\mathbb{E}_{XY}\Tr[\rho_Y^i\phi_X].$$
By Theorem~\ref{thm:main_exist}, forall $Y$ such that $H_{\min}(X|Y)\geq \ell$, $i\in[t]$, $\mathbb{E}_{XY}\Tr[\rho_Y^i\phi_X]=\negl(n)$ given that $\phi$ is a quantum cryptographic hash function. $e_S^{V*}(\ell)=\negl(n)$ then comes from that $t=\omega(\log n)$.
\end{proof}

\section{Applications of Separation Lemma in Quantum-proof Extractors}
\label{sec:re}
Let us now recall the setting of seeded extractor. A seeded extractor $Ext:\{0,1\}^n\times \{0,1\}^d\rightarrow \{0,1\}^m$ takes a weak random source $X$ as well as a much shorter, uniform and independent seed $Y$ and outputs a nearly uniform distribution. Rigorously we have the following definition.
\begin{definition}[Extractor]
    A function $Ext:\{0,1\}^n\times\{0,1\}^d\rightarrow \{0,1\}^m$ is called an \emph{$(k,\epsilon)$-extractor} if for all random source $X$ such that $-\log p_g(X)\geq k$, we have
    $$\Vert U_m-Ext(X\otimes U_d)\Vert\leq \epsilon.$$
\end{definition}

When used in the setting of privacy amplification, one would consider the case where there is a leakage of the random source. The output of the extractor then need to not only be close to uniform, but also be almost independent of the side information. Depending on whether the leakage is classical or quantum, we have the following definitions for \emph{classical-proof} and \emph{quanutm-proof} extractors respectively.

\begin{definition}[Classical-proof Extractor]
    A function $Ext:\{0,1\}^n\times\{0,1\}^d\rightarrow \{0,1\}^m$ is called an \emph{$(k,\epsilon)$-classical-proof extractor} if for all random source $X$ with classical side information $Y$ such that $-\log p_g(X|Y)\geq k$, we have
    $$\Vert U_m\otimes Y-Ext(X\otimes U_d)Y\Vert\leq \epsilon.$$
\end{definition}

\begin{definition}[Quantum-proof Extractor]
    A function $Ext:\{0,1\}^n\times\{0,1\}^d\rightarrow \{0,1\}^m$ is called an \emph{$(k,\epsilon)$-quantum-proof extractor} if for every cq state $\rho_{XE}$  such that $-\log p_g(X|E)\geq k$, we have
    $$\Vert U_m\otimes \rho_E-(Ext\otimes \mathcal{I})(\rho_{XE})\Vert_{tr}\leq \epsilon.$$
\end{definition}

Intensive research on all three kinds of extractors has been done over years. Fortunately, the following theorem says that a good extractor is automatically a good classical-proof extractor, up to very little parameter loss:

\begin{theorem}[~\cite{konig2008bounded}]
    \label{thm:cpext}
    Any $(k,\epsilon)$-extractor is a $(k+\log 1/\epsilon, 2\epsilon)$-classical-proof extractor.
\end{theorem}

One long-standing open problem, then, is whether a classical-proof extractor is essentially a quantum-proof extractor with roughly the same parameter. Currently the best result is that a $(k,\epsilon)$-classical proof extractor is a $(k+\log 2/\epsilon, O(2^{m/2}\sqrt{\epsilon}))$-extractor. For nonexponential blow-up of parameters, Gavinsky et al.~\cite{gavinsky2007exponential} showed that there exists a
$(k-\Theta(n),\epsilon)$-classical extractor which is not secure against $O(\log n)$ qubits of quantum side information.  Proving a result without exponential blowup on the parameter in the practical range, however, is a very challenging problem. 

The Separation Lemma here suggests that such a result without a drastic blowing up on the error parameter may not exist. To see this, we need to define some sets first.

\begin{definition}
    For every $k$, define the following sets:
    \begin{itemize}
        \item $C(k)$ is the set of all cq states $\rho_{XE}$ that can be generated from a classical random source $XY$ with $$-\log \max_{y}p_g(X|Y=y)\geq k$$ via a channel acting only on the $Y$ part;
        \item $CC(k)$ is the set of all cq states $\rho_{XE}$ that can be generated from a classical random source $XY$ with $$-\log p_g(X|Y)\geq k$$ via a channel acting only on the $Y$ part;
        \item $CQ(k)$ is the set of all cq states $\rho_{XE}$ such that $-\log p_g(X|E)_\rho\geq k$.
        \end{itemize}
\end{definition}

Clearly we have $C(k)\subseteq CC(k)\subseteq CQ(k)$. To see the importance of these sets, we can use the alternative, though equivalent definitions of extractors, classical-proof and quantum-proof extractors:
\begin{definition}[Alternative definition for extractors]
    A function $Ext:\{0,1\}^n\times \{0,1\}^d\rightarrow\{0,1\}^m$ is called a $(k,\epsilon)$-extractor (classical-proof extractor, quantum extractor), if for all cq states $\rho_{XE}\in C(k) (CC(k), CQ(k))$ with $X\in\{0,1\}^n$ we have
    $$\Vert U_m\otimes\rho_E-(Ext\otimes \mathcal{I})(\rho_{XE})\Vert_{tr}\leq \epsilon.$$
\end{definition}

The proof of Theorem~\ref{thm:cpext} essentially makes use of the fact that every state $\rho_{XE}$ in $CC(k+\log 1/\epsilon)$ is $\epsilon$-close to the set $C(k)$, and thus a $(k,\epsilon)$-extractor, when applied to $\rho$, would introduce at most $2\epsilon$ error from the state $U_m\otimes\rho_E$ due to monoticity of trace distance. If one could obtain similar results between $CQ(k)$ and $CC(k)$, then we could easily prove that a classical-proof extractor is secret quantum-proof.

This turns out to not be the case according to Separation Lemma. Recall our definition of $\epsilon$-smooth conversion parameter $p_{\downarrow}^{\epsilon}(X|E)_\rho$, which measures the least guessing probability of a classical distribution we need in order to generate the desired state $\rho$, i.e.\
$$d(\rho, CC(k))\leq \epsilon \Leftrightarrow -\log p_{\downarrow}^\epsilon(X|E)_\rho\geq k.$$

\begin{theorem}
    \label{thm:topsep}
    For any $k$, there exists a cq state $\rho_{XE}$ where $X\in\{0,1\}^{k(1+o(1))}$ such that $\rho\in CQ(k)$ while $d(\rho, CC(2))\geq 0.98$.
\end{theorem}
\begin{proof}

Fix $\delta = 0.02$ and $\epsilon = 0.98$. By the Johnson-Lindenstrauss Lemma, we have a cq state $$\rho_{XE}=2^{-n}\sum_{x\in\{0,1\}^n}|x\rangle\langle x|\otimes|\psi_x\rangle\langle \psi_x|\in\mathcal{S}_{=}(\X\otimes \E)$$
such that $\log \dim \E=O(\log n)$ and $\max_{x\neq x'}|\langle \psi_x|\psi_{x'}\rangle|$. For given $k$, there exists $n=k(1+o(1))$ such that $n-O(\log n)\geq k$. With that $n$, we have 
$$-\log p_g(X|E)_\rho\geq k\Rightarrow \rho\in CQ(k)$$ as well as $$-\log p^{\epsilon}_{\downarrow}(X|E)_{\rho}\leq 2.$$ By the Separation Lemma, we have 
$$p_{\downarrow}^{\epsilon}(X|E)_\rho \geq 1-\delta-\epsilon/2\Rightarrow d(\rho,CC(2))\geq 0.98.$$ 
\end{proof}

The Separation Lemma suggests that an arbitrary classical extractor may not be quantum-proof, i.e.\ the sets $CQ(k)$ and $CC(k)$ are spatially separated. Nevertheless, one may still prove that a classical extractor is quantum-proof, but in order to do that, one might need to use additional properties of the extractor, other than that it can extract randomness from all states in $CC(k)$.

\commentout{
\section{Conclusion}
In this work, we proved that there exists a quantum hash scheme which is resilient against leakage of classical side information to any information-theoretic adversarial party. Still, one may wonder what will happen in the context that the leakage can also be quantum. Although no quantum hash can be secure against an arbitrary quantum leakage, (the worst case being that the leakage is the quantum hash itself), one can consider more restricted situation, where the quantum leakage is local with respect to the classical information, or the adversarial
party has limited computational power or bounded storage, etc.

In the context of randomness extractor, the problem that whether an extractor is quantum-proof with comparable parameters is yet to be proved. Proving that a classical-proof extractor is in general quantum-proof seems now promising if some argument on spatial separation between $CC(k)$ and $CQ(k)$ can be proven. Nevertheless, there is still possibility that we can get much better parameter loss compared to existing result; and even if such a result does not exists, it is still possible that there exists a quantum-proof extractor with parameters comparable to known classical constructions.
}
\paragraph{Acknowledgements.} The authors would like to thank Kai-Min Chung, Ronald de Wolf, Carl Miller, Michael Newman and Fang Zhang for useful discussion on this topic. This research was supported in part by US NSF Award 1216729.
\newpage
\appendix
    \section{Proof of the Cotlar-Stein Lemma}
    \label{sec:csl}

\begin{lemma}[Cotlar-Stein Lemma]
\label{lem:csl2}
For a set of unit vectors $\{|\psi_1\rangle,|\psi_2\rangle,\cdots,|\psi_n\rangle\}$ with maximum fidelity $\max_{i,j:i\neq j}|\langle \psi_i|\psi_j\rangle|\leq \delta$, we have
$$\lambda_{\max}\left(\sum_{i=1}^n|\psi_i\rangle\langle \psi_i|\right)\leq 1+(n-1)\delta.$$
\end{lemma}

\begin{proof}
We use the fact that the operator norm is upper bounded by all Schatten $p$ norms, i.e.

$$\lambda_{\max}(\rho)=\Vert \rho\Vert_\infty=\lim_{p\rightarrow \infty}\left(\Tr[\rho^p]\right)^{1/p}.$$

For an arbitrary positive integer $m$, let's now bound $$\Vert \sum_{i=1}^n|\psi_i\rangle\langle \psi_i|\Vert_m^m=\Tr\left[(\sum_{i=1}^n|\psi_i\rangle\langle\psi_i|)^m\right].$$ We have
\begin{align*}
\Tr\left[(\sum_{i=1}^n|\psi_i\rangle\langle\psi_i|)^m\right]&=\sum_{i_1,\cdots, i_m\in[n]}\Tr\left[\prod_{j=1}^m |\psi_{i_j}\rangle\langle \psi_{i_j}|\right]\\
&=\sum_{i_1,\cdots, i_m\in[n]}\prod_{j=1}^{m-1}\langle \psi_{i_j}|\psi_{i_{j+1}}\rangle\cdot \langle \psi_{i_m}|\psi_{i_1}\rangle\\
&\leq \sum_{i_1,\cdots, i_m\in[n]}\prod_{j=1}^{m-1}|\langle \psi_{i_j}|\psi_{i_{j+1}}\rangle|\cdot |\langle\psi_{i_m}|\psi_{i_1}\rangle|.\\
\end{align*}
Using the fact that both $|\psi_{i_m}\rangle$ and $|\psi_{i_1}\rangle$ are unit vectors, we have $|\langle \psi_{i_m}|\psi_{1}\rangle|\leq 1$. Then
\begin{align*}
 \Tr\left[(\sum_{i=1}^n|\psi_i\rangle\langle\psi_i|)^m\right]&\leq \sum_{i_1,\cdots, i_m\in[n]}\prod_{j=1}^{m-1}|\langle \psi_{i_j}|\psi_{i_{j+1}}\rangle|\\
&\leq \sum_{i_1,\cdots, i_{m-1}\in[n]}\prod_{j=1}^{m-2}|\langle \psi_{i_j}|\psi_{i_{j+1}}\rangle|\cdot\sum_{i_m}|\langle\psi_{i_{m-1}}|\psi_{i_{m}}\rangle|\\
\end{align*}
Note that for every $i_{m-1}$, the term $\sum_{i_{m}}|\langle \psi_{i_{m-1}}|\psi_{i_{m}}\rangle|$ can be upper bounded by $1+(n-1)\delta$. Repeatedly applying this argument, we have
\begin{align*}
\Tr[(\sum_{i=1}^n|\psi_i\rangle\langle\psi_i|)^m]&\leq \sum_{i_1,\cdots, i_{m-1}\in[n]}\prod_{j=1}^{m-2}|\langle \psi_{i_j}|\psi_{i_{j+1}}\rangle|\cdot (1+(n-1)\delta)\\
&\leq \sum_{i_1,\cdots, i_{m-2}\in[n]}\prod_{j=1}^{m-3}|\langle \psi_{i_j}|\psi_{i_{j+1}}\rangle|\cdot (1+(n-1)\delta)^2\\
&\leq \cdots\\
&\leq\sum_{i_1}(1+(n-1)\delta)^{m-1}\\
&= n \cdot (1+(n-1)\delta)^{m-1}.
\end{align*}
Therefore, for every $m$ we have $$\lambda_{\max}(\sum_{i=1}^n|\psi_i\rangle\langle\psi_i|)\leq (1+(n-1)\delta)^{1-\frac{1}{m}}\cdot n^{1/m}.$$
The result follows by letting $m\rightarrow \infty$.
\end{proof}